\documentclass[12pt]{article}
\usepackage{setspace}
\usepackage{color}
\usepackage{amsfonts}
\usepackage{amssymb}
\usepackage{amsmath}
\usepackage[mathscr]{eucal}
\usepackage{amsfonts}
\usepackage{amsmath,amsxtra,latexsym,amsthm, amssymb, amscd}
\usepackage[colorlinks,citecolor=blue,filecolor=black,linkcolor=blue,urlcolor=black]{hyperref}
\usepackage{pgfplots} 
\usepackage{url}
\usepackage[round]{natbib}
\usepackage{fancyhdr}

 \usepackage{natbib}
\usepackage{graphicx}
\usepackage{hyperref}

\usepackage{caption}
\usepackage{epstopdf}

\usepackage{soul}

\newtheorem{theorem}{Theorem}

\newtheorem{definition}[theorem]{Definition}

\newtheorem{Simulation}{Simulation}

\newtheorem{lemma}{Lemma}

\newtheorem{proposition}{Proposition}
\newtheorem{remark}{Remark}

\newtheorem{assum}{Assumption}

\newcommand{\rr}{\mathbb{R}}

\begin{document}
\title{Governance, productivity and economic development\footnote{We would like to thank an associate editor and two anonymous referees for providing valuable comments and suggestions.}
}
\author{Cuong LE VAN\footnote{IPAG, CNRS, PSE, and TIMAS ThangLong University. Email: Cuong.Le-Van@univ-paris1.fr.} \and Ngoc-Sang PHAM\thanks{EM Normandie Business School, M\'etis Lab. Email: npham@em-normandie.fr.} \and Thi Kim Cuong PHAM\footnote{\emph{Corresponding author}. EconomiX,  Paris Nanterre University and VNU University of Economics and Business, Vietnam National University, Hanoi. Email: pham\_tkc@parisnanterre.fr. Phone: +33 1 40 97 77 21. Address: B\^atiment G - Maurice Allais, 200, Avenue de la R\'epublique, 92001 Nanterre cedex} \and Binh TRAN-NAM\footnote{UNSW Business School, UNSW Sydney.  Email: b.tran-nam@unsw.edu.au}}

\date{\today}
\maketitle

\begin{abstract}
This paper explores the interplay between transfer policies, R\&D, corruption, and economic development using a general equilibrium model with heterogeneous agents and a government. The government collects taxes, redistributes fiscal revenues, and undertakes public investment (in R\&D, infrastructure, etc.). Corruption is modeled as a fraction of tax revenues that is siphoned off and removed from the economy. We first establish the existence of a political-economic equilibrium. Then, using an analytically tractable framework with two private agents, we examine the effects of corruption and evaluate the impact of various policies, including redistribution and innovation-led strategies.\\
\newline


\end{abstract}

\section{Introduction}

Corruption is one of the most classical topics, not only in economics but also in many other fields.\footnote{See, for example, \cite{op12} for a detailed survey in the context of developing countries.} 
The Transparency International defines ``corruption as the abuse of entrusted power for private gain'' (Transparency International (2025)). Of course, corruption is a multidimensional concept. A large number of studies, both theoretical and empirical, show that corruption has a harmful effect on economic performance. 
For instance, \cite{Mauro95} provides a cross-country empirical analysis during the period 1980-1983, which shows a significant negative association between corruption and investment, as well as growth. The corruption variable used in this empirical study corresponds to ``the degree to which business transactions involve corruption or questionable payments'' from the Business International Corporation definition.


\cite{go19} present empirical evidence based on data for 175 over the period 2012–2018 – a period during which the Corruption Perceptions Index (CPI) is comparable
across countries and over time. They show that corruption (i.e., CPI)  is negatively associated with economic growth, especially in autocracies and countries with low government effectiveness and rule of law. Similarly, using a panel of 142 countries from 1994 to 2014, \cite{cg18} shows that higher corruption is strongly associated with lower GDP growth and foreign investment ratio. 

Other empirical studies further indicate the adverse consequences of corruption on other dimensions of an economy, such as human capital accumulation and environmental quality. As underlined in \cite{Abdulla2021}, corruption has a negative effect on the stock of human capital and its elimination would increase aggregate output by 18-21\% on average. In parallel, \cite{wdzw18} analyze the relationship between economic growth, CO2 emissions, and corruption in a panel of BRICS countries from 1996 to 2015, concluding that control of corruption is crucial for mediating the relationship between economic growth and environmental quality and can contribute to reduction in CO2 emissions. In an analysis using a panel covering 21 Central and Eastern European countries, \cite{Petrova2020} finds that higher levels of corruption and bureacratic inefficiency are associated with lower levels of redistributive spending.  

Some theoretical studies model the impact of corruption of economic outcomes. \cite{aacc16} develop a growth model with innovation to analyze the effect of taxation and corruption on growth and innovation. They demonstrate that corruption lowers the optimal tax rate and reduces growth. Moreover, their empirical analysis across US states confirms that higher local corruption attenuates the positive effect of taxation on both growth and innovation, suggesting that if governments are aiming for economic growth, investing resources in fighting corruption makes a lot of sense.  \cite{mv20} analyze the link between corruption, economic growth, and inflation into a monetary growth model where a corruption sector allows households evading from taxation. Their finding predicts a U-shaped relationship between corruption and inflation, whereby beyond a threshold value of corruption, corruption and inflation move in the same direction
and lower the efficiency of seigniorage.


Conversely, a strand of literature suggests that, under some conditions, corruption and rent-seeking may potentially be associated with positive outcomes, in particular when corruption contributes to ``grease the wheels''. The seminar paper of \cite{Leff64} discusses a particular type of corruption: the practice of buying favors from the bureaucrats responsible for formulating and administering the government’s economic policies.   It is unsurprising that this kind of market activities may be beneficial, which leads to the  “grease the wheels” hypothesis. 
More recently, \cite{mw10}, using a panel of 69  both developed and developing countries when exploring the impact of corruption on aggregate efficiency, suggest that corruption (measured by CPI from Transparency International or the World Bank's corruption indicator) is less detrimental in countries where the rest of the institutional framework is weaker. In particular, in some countries where institutions are particularly ineffective, corruption may serve as a ``grease to the wheels''. 
At the local level, \cite{dhls22} highlight the importance of firm and city heterogeneity in shaping firms' productivity reaction to corruption, in particular, for firms that are export-oriented, more profitable, publicly owned, grow-faster, the corruption may positively impact their productivity. A recent study of \cite{hs25} focuses on the relationship between corruption and income inequality. This analysis, using a dataset covering up to 160 countries, concludes that corruption does not necessarily increase income inequality. If the objective of public policy is to reduce income inequality, anti-corruption efforts is not sufficient, but targeting unemployment, a robust driver of inequality, should be prioritized in government interventions. There are also empirical studies suggesting that tax corruption may, at least in the short run, promote innovation. For example, a study by \cite{dvtn21}, using firm-level panel data from Vietnam covering the period 2005–2015, finds that petty tax corruption has a positive effect on all types of firm-level innovative activities. One possible explanation is that firms may use the ``tax savings'' from corruption to finance business improvements, including various types of innovative inputs.

Despite these insights, few studies have simultaneously examined the interplay between corruption, transfer policies, innovation and  economic development, specially within a general-equilibrium framework. This paper develops then a two-period general equilibrium model to address the following key questions:
\begin{enumerate}
\item What are the macroeconomic effects of the interaction between taxation, redistribution and public investment in R\&D? Could they improve economic outcomes even in the presence of corruption? 

\item Are the economic outcomes in the case without intervention better than those under intervention when corruption is present? May the co-existence of corruption and innovation be Pareto-improving? 
\end{enumerate}


Our framework considers a finite number of heterogeneous agents who can borrow/lend through a financial market and produce a single output of the economy. She has her own production function. Each agent must pay a tax that equals a fraction of her income and receives a transfer from the government. Both tax rates and transfers are individualized and time-dependent. In the first period, the government also makes a public investment (including R\&D, infrastructure, education, etc.), which will improve agents' productivity in the second period. However, there may exist a fraction of collected tax revenue that disappears and fails to contribute to economic activity. We refer to this lost fraction as {\it corruption}.


Our paper makes three main contributions. Our first one (Proposition \ref{prop22}) is to establish the existence of a political-economic equilibrium with externality related to transfers and innovation. Our proof consists of two steps: i) given a vector of government tax revenues $\mathcal{T}$, we prove, under mild assumptions, the existence and uniqueness of a competitive equilibrium. This equilibrium is continuous in $\mathcal{T}$. Moreover, it generates tax revenue $T$ for the government; ii)  using the Brouwer fixed point theorem, we prove the existence of a vector $\mathcal{T}^*$ so that the equilibrium associated with this vector generates the government's revenue which coincides with $\mathcal{T}^*$, i.e., $T=\mathcal{T}^*$. So, there exists a political-economic equilibrium with transfers and innovation.\footnote{Our proof of the equilibrium existence is inspired by \cite{mitra98}, 
\cite{levan02}, \cite{levan04}, \cite{ab06}, \cite{blp25}. The main difference is that they work under optimal growth models (with a representative consumer) while we work with a two-period model with a finite number of consumers.}



Our second contribution is to explore the effects of various redistribution policies and public investment, as well as the impact of corruption, within a tractable framework that features individualized taxes and transfers. After computing the equilibrium outcomes, we show that public investment and R\&D efficiency have a positive impact on aggregate output, while corruption exerts a negative effect.

Our third contribution is to provide policy insight on redistribution and R\&D policy by comparing the GDPs in scenario S0. inaction (i.e., no government intervention) with one of the three following cases: S1. intervention with public investment in R\&D and redistribution policy; S2. intervention with only redistribution policy; and S3. distorted taxation, a very high tax on the most productive agent.

\begin{table}[h!]
\centering
\begin{tabular}{|p{3.5cm}|p{3cm}|p{2.5cm}|p{4.5cm}|}
\hline
\textbf{Scenario} & \textbf{Government Action} & \textbf{Corruption} & \textbf{Outcomes}\\
\hline
\textbf{S0}. Inaction & No intervention & None & Baseline outcome \\
\hline
\textbf{S1}. Public investment in R\&D and redistribution policy  & High public investment & Present or not & If investment is efficient, productivity increases; output rises,  even with corruption \\
\hline
\textbf{S2}. Targeted redistribution policy & Redistribution from low to high productivity agents & Present or not & Transfers boost investment where returns are highest, leading to higher aggregate output, even with corruption. \\
\hline
\textbf{S3}. Distorted taxation & High taxes on productive or high-investment agents & Present or not & Discourages investment; reduces aggregate investment and output. \\
\hline
\end{tabular}
\caption{Comparison of economic outcomes under different government intervention scenarios}
\label{tab:intervention_comparison}
\end{table}

While we prove that the corruption is always harmful, we argue that inaction may be worse than imperfect interventions for several reasons. Indeed, as shown in Table 1, we highlight two scenarios (S1, S2) in which imperfect interventions (i.e., with corruption) lead to better outcomes than the inaction scenario (S0). In scenario S1, when the government engages in significant public investment and the efficiency of that investment is high (in the sense that it meaningfully enhances firm productivity), firms become more productive and, thanks to this, the output will be higher than that in the inaction scenario (S0), even if a small fraction of the output is lost to corruption. In scenario S2, the government taxes agents with low productivity and low discount factors and transfers resources to the most productive agents. In this case, aggregate investment and output would increase, even in the presence of pretty corruption.

These insights enable us to understand why corruption and relatively good outcomes can coexist, as observed in some empirical studies. By the way, our theoretical results contribute to clarifying the ``grease to the wheels'' hypothesis. However, we also show that poorly designed redistribution policies, such as imposing high taxes on agents having high productivity or high investment rates (scenario S3, distorted taxation), can be detrimental, as they may reduce aggregate investment and overall production.\footnote{In reality, it is not easy to determine whether a distorted tax results from corruption or not.} 




From a theoretical point of view, our paper is related to \cite{dl02}. There are, however, several key differences. First, \cite{dl02} studies an optimal growth model (with one representative agent) while our model has heterogeneous agents. Second,  \cite{dl02} considers a corruption regarding international aid (in the form of a loan) 
while we study corruption via the redistribution process.\footnote{In \cite{dl02}, corruption means that a fraction of aid is diverted by some bureaucrats. This fraction can disappear from the country or come back to the economy as an extra-consumption or extra-investment. In our paper, corruption means that a fraction of tax revenue disappears from the country. \cite{dl02} show that international aid may be beneficial for economic growth even the corruption takes place if international aid (or corrupted amount from aid) is used to increase the aggregate investment or the incentive of private firms).}  

The remainder of the paper is organized as follows. Section \ref{model} presents a general equilibrium model with public investment and corruption.  Section \ref{existence} investigates the existence of equilibrium, while Section  \ref{analyses} addresses the effects of different policies and corruption. Section  \ref{conclusion} concludes. Technical proofs are presented in Appendix \ref{appendix}.

\section{A model with transfers, public investment, and corruption}
\label{model}
\subsection{Individual choice}

We consider a two-period economy with $m$ heterogenous agents ($i=1,\ldots,m$)  and a government.  There is one good. At date $0$, the agent $i$ is endowed $w_{i,0}$ units of good. She has to pay a tax $\tau_{i,0} w_{i,0}$ and also receives $\gamma_{i,0}T_0$ from the government. The net transfer received by agent $i$ is $\gamma_{i,0}T_0-\tau_{i,0} w_{i,0}$ which can be positive or negative.  We require that $\tau_{i,0},\gamma_{i,0}$ are in the interval $[0,1]$.

Each agent can borrow or lend an amount $ b_{i,0} $ at date $0$ with the real return $R_1$, which is endogenous. Assume that there is no borrowing constraint. She can invest in capital with an amount of $k_{i,1}$ to produce at date $1$ following the technology which is characterized by the function $A_iF_i(k_{i,1})$ where $A_i$ represents the productivity or the technological level. 

At date $0$, each agent chooses current and future consumption ($c_{i,0},c_{i,1}$), capital investment ($k_{i,1}$), and saving or borrowing ($ b_{i,0} $) to maximize her intertemporal utility. The problem of agent $i$ can be written as:
\begin{subequations}\label{agentiproblem}
\begin{align}
&\max_{(c_{i,0},c_{i,1}, k_{i,1}, b_{i,0} )}u_i(c_{i,0})+\beta_iu(c_{i,1})\\
\text{subject to constraints: } &c_{i,0}+k_{i,1}\leq (1-\tau_{i,0})w_{i,0}+ b_{i,0} +\gamma_{i,0}T_0\\
&c_{i,1}\leq (1-\tau_{i,1})A_iF_i(k_{i,1})-R_1 b_{i,0} +\gamma_{i,1}T_1\\
&c_{i,0}\geq 0, c_{i,1}\geq 0, k_{i,1}\geq 0
\end{align}
\end{subequations}
where $\beta_i\in (0,1)$ is the rate of time preference of agent $i$. 

The redistribution policy at date $1$ is represented by the rates $\tau_{i,1}$ and $\gamma_{i,1}$ which are also in the interval $[0,1)$.

We assume that the productivity $A_i$ depends on the public investment in the sense that $${A_i=\mathcal{A}_i(D_0)}$$ where $D_0$ represents the public investment which includes investment in R\&D, education, public infrastructure, etc. (\cite{pp20}). We can relate individual productivity $A_i$ to innovation induced by public investment. The amount of public investment $D_0$ is taken as given by any agent.

We require the following basic assumption about the function $\mathcal{A}(\cdot)$.
\begin{assum}\label{assum_Ai}
The function $\mathcal{A}(\cdot): \rr_+\to \rr_+$ is continuous and increasing. $\mathcal{A}_i(0)>0$.
\end{assum}
This assumption implies that without public investment, the production function of agent $i$ becomes $ \mathcal{A}_i(0)F_i(k_{i,1}) $.


\subsection{Government}
At date $0$, the government collects an amount $\tau_{i,0} w_{i,0}  \equiv T_0$ of good from each agent $i$ and redistributes it in the form of transfers $\gamma_{i,0}T_0$. Moreover, the government spends $D_0=\gamma_{d,0}T_0$ in  public investment.

We note that individuals can borrow or lend in the first period 0, while the government cannot.  As a result, government budgets are balanced in both periods. Its budget constraint at date $0$ is written as: 
\begin{align}
\sum_{i=1}^m\tau_{i,0}w_{i,0}=T_0=G_0=\sum_{i=1}^m\gamma_{i,0}T_0+\gamma_{d,0}T_0+(1-\sum_{i=1}^m\gamma_{i,0}-\gamma_{d,0})T_0
\end{align}

At date $1$, the government collects an amount $\tau_{i,1}A_iF_i(k_{i,1})$ from each agent and does the redistribution. Since we consider a two-period model, we abstract from R\&D in the second period.
\begin{align}
\sum_{i=1}^m\tau_{i,1}A_iF_i(k_{i,1})=T_1=G_1=\sum_{i=1}^m\gamma_{i,1}T_1+(1-\sum_{i=1}^m\gamma_{i,1})T_1.
\end{align}



The amounts $(1-\sum_{i=1}^m\gamma_{i,0}-\gamma_{d,0})T_0$ and $(1- \sum_{i=1} \gamma_{i,1})T_1$ are retained by the government. These amounts are not returned to the economy, reflecting the presence of corruption.\footnote{Alternatively, we can view the disappeared public resources as governance inefficiency or public transaction costs of government spending.}

Let us clarify some notions.
\begin{definition}The list $(\tau_{i,t},\gamma_{i,t})_{i=1}^m$ is the redistribution policy of the government at date $t$. The amount $\gamma_{d,0}T_0$ represents the government effort in  R\&D  at date $0$. We say that 
\begin{itemize}
   \item[-] There is a corruption at date $0$ if the fraction $\gamma_{c,0}\equiv (1-\sum_{i=1}^m\gamma_{i,0}-\gamma_{d,0})$ is strictly positive.  
\item[-] There is a corruption at date $1$ if the fraction $\gamma_{c,1}\equiv 1-\sum_{i=1}^m\gamma_{i,1}$ is strictly positive. 
\item[-]There is no intervention if taxes, transfers and R\&D are zero, i.e. $\tau_{i,t}=\gamma_{i,t}=\gamma_{d,t}=0$ for any $i$ and for any $t=0,1$.
\end{itemize}
\end{definition}

Precisely, we consider that corruption exists if a part of the resources disappears from the economy. This part of the resources corresponds to the positive fraction $1- \sum_{i=1} \gamma_{i,0} - \gamma_{d,0}$ at date 0 and  $1- \sum_{i=1} \gamma_{i,1}$ at date 1. We note that other papers add diverted resources to household budget as a corruption consumption (\cite{dl02}) or corruption income (\cite{aacc16}). 

It should be noticed that when there is no intervention, there is no corruption, and the government is not analyzed in this economy. Our modeling is similar to a recent empirical study, which shows that corruption may distort the public spending structure (\cite{ma25}). When analyzing the impact of corruption on the allocation of public expenditures across 45 sub-Saharan African countries, these authors show that corruption leads to a significant decrease in capital expenditures, diverting resources away from long-term investment projects. 

We impose the following assumptions on utility, production functions, and other constraints.
\begin{assum}\label{assum-agenti}$w_{i,0}>0, \tau_{i,t}\in [0,1), \gamma_{i,t},\gamma_{d,0}\in [0,1]$, $\gamma_{d,1} = 0$, and $\sum_{i=1}^m\gamma_{i,t}+\gamma_{d,t}\leq 1$ for any $i$, for any $t$.

For any $i$, the utility function $u_i$ is twice continuously differentiable, strictly increasing, strictly concave and $u_i'(0)=\infty$ while the production function $F_i$ is continuously differentiable, concave, strictly increasing, and $F_i(0)=0$.
\end{assum}

\subsection{General equilibrium}



We now present our notion of equilibrium.

\begin{definition}[political-economic equilibrium]%
\label{def-equilibrium-2} A political-economic equilibrium with transfers and public investment is a non-negative list 

$$\Big((c_{i,0},c_{i,1},k_{i,1}, b_{i,0} )_{i=1}^m, R_1,T_0,T_1,G_0, G_1, D_0 \Big)$$ satisfying the following conditions:

\begin{enumerate}
\item $R_1>0$.
\item\label{defpoint2} $G_0=T_0=\sum_{i=1}^m\tau_{i,0} w_{i,0},$ $G_1=T_1=\sum_{i=1}^m\tau_{i,1}\mathcal{A}_i(D_0)F_i(k_{i,1})$, $D_0=\gamma_{d,0} T_0$.
\item Given $(R_1, T_0,T_1,G_0, G_1, D_0)$, the allocation $(c_{i,0},c_{i,1},k_{i,1}, b_{i,0} )$ is a solution to the problem of agent $i$.

\item Market-clearing conditions:
\begin{align}
\sum_{i=1}^m(c_{i,0}+k_{i,1})&=\sum_{i=1}^m(1-\tau_{i,0})w_{i,0}+\sum_{i=1}^m\gamma_{i,0} T_0\\
\sum_{i=1}^mc_{i,1}&=\sum_{i=1}^m(1-\tau_{i,1})\mathcal{A}_i(D_0)F_{i}(k_{i,1})+\sum_{i=1}^m\gamma_{i,1}T_1\\
\sum_{i=1}^m b_{i,0} &=0.
\end{align} 
\end{enumerate}
\end{definition}

\section{Existence of equilibrium}
\label{existence}

The equilibrium defined in Definition \ref{def-equilibrium-2} involves externalities and endogenous transfers, which makes the task of establishing its existence non-trivial. This section aims to prove its existence in two steps.  First, given government tax revenues $T_0$ and $T_1$, we prove the existence and uniqueness of a competitive equilibrium. This equilibrium is continuous in $(T_0, T_1)$. Moreover, it generates tax revenues for the government. Secondly,  using the Brouwer fixed point theorem, we prove the existence of a couple ($T_0^*, T_1^*$) so that the equilibrium associated with this couple generates the government's revenues, which are exactly ($T_0^*, T_1^*$). So, there exists a political-economic equilibrium with transfers and innovation in Definition \ref{def-equilibrium-2}.

\paragraph{Step 1.} Let us start by introducing the notion of the competitive equilibrium given government intervention $(T_0, T_1, D_0)$.

\begin{definition}[competitive equilibrium]\label{def-equilibrium-0} Let $(T_0,T_1,D_0)\in \rr_+^3$ be given and $D_0\leq T_0$. A competitive  equilibrium is a list $\big((c_{i,0},c_{i,1},k_{i,1}, b_{i,0} )_{i=1}^m, R_1\big)$ satisfying the following conditions:
\begin{enumerate}
\item $R_1>0$.
\item Given $R_1$, the allocation $(c_{i,0},c_{i,1},k_{i,1}, b_{i,0} )$ is a solution to the agent i's  problem.

\item Market-clearing conditions:
\begin{align}
\sum_{i=1}^m(c_{i,0}+k_{i,1})&=\sum_{i=1}^m(1-\tau_{i,0})w_{i,0}+\big(\sum_{i=1}^m\gamma_{i,0}\big)T_0\\
\sum_{i=1}^mc_{i,1}&=\sum_{i=1}^m(1-\tau_{i,1})\mathcal{A}_i(D_0)F_{i}(k_{i,1})+\big(\sum_{i=1}^m\gamma_{i,1}\big)T_1\\
\sum_{i=1}^m b_{i,0} &=0.
\end{align} 
\end{enumerate}
\end{definition}
\begin{remark}\label{remark2}
Denote $ W_0\equiv \sum_{i=1}^mw_{i,0}$ the aggregate endowment at date $0$ and $$\bar{T}_1\equiv \max_{(k_1,\ldots,k_m)\in \rr_+^m: \sum_ik_i\leq W_0 }\tau_{i,1}\mathcal{A}_i(\gamma_{d,0}W_0)F_i(k_{i,1}),$$
which is an upper bound of the tax revenue at date $1$. Then, we define  $\bar{W}\equiv \max(\sum_{i=1}^m\tau_{i,0} w_{i,0},\bar{T}_1)$.

In any political-economic equilibrium (Definition \ref{def-equilibrium-2}), according to the market-clearing conditions, we observe that 
\begin{subequations}
\begin{align}
T_0&=\sum_{i=1}^m\tau_{i,0} w_{i,0}\leq  \bar{W}\\
 \sum_{i}k_{i,1}&\leq \sum_{i=1}^m(c_{i,0}+k_{i,1})=\sum_{i=1}^m(1-\tau_{i,0})w_{i,0}+\big(\sum_{i=1}^m\gamma_{i,0}\big)T_0\leq W_0\leq \bar{W} \\
T_1&=\sum_{i=1}^m\tau_{i,1}\mathcal{A}_i(D_0)F_i(k_{i,1})= \sum_{i=1}^m\tau_{i,1}\mathcal{A}_i(\gamma_{d,0}W_0)F_i(k_{i,1}) \leq \bar{T}_1\leq \bar{W}.
\end{align}
\end{subequations}
\end{remark}

 For each $(T_0,T_1)\in [0,\bar{W}]^2$ and $D_0=\gamma_{d,0}T_0$ where $\gamma_d\in [0,1]$, let us denote $\mathcal{E}_1(T_0,T_1)$ the set of competitive equilibria. By using the standard argument in the literature (see \cite{lvp16} for instance), we can prove that the equilibrium set $\mathcal{E}_1(T_0, T_1)$ is not empty.

 \begin{lemma}\label{lemmaexistence}Under Assumptions  \ref{assum_Ai} and \ref{assum-agenti}, there exists a competitive equilibrium. 
\end{lemma}

It is well known (see \cite{mc95}, Chapter 17, for instance) that without additional assumptions, the competitive equilibrium may not be unique. To obtain the uniqueness of competitive equilibrium, we require the following assumption. 
\begin{assum}\label{assumption-additional} Assume that, for any $i$, $F_i$ is strictly concave, $F_i'(0)=\infty$ and $u_i^{\prime}(c)+cu_i^{\prime \prime}(c)\geq 0$ for any $c>0$.
 \end{assum}

\begin{proposition}\label{prop12}Let Assumptions  \ref{assum_Ai}, \ref{assum-agenti}, and \ref{assumption-additional}  be satisfied.  Then there exists a unique competitive equilibrium, and it is continuous in $(T_0,T_1,D_0)$.
\end{proposition}
\begin{proof}
See Appendix \ref{appendix}.\end{proof}
Let us explain the intuition in Proposition \ref{prop12}. Condition $F_i'(0)=\infty$ ensures that $R_1=(1-\tau_{i,1})A_iF_i'(k_{i,1})$. Since $F_i$ is strictly concave, $k_{i,1}$ is a differentiable function of $R_1$.\footnote{With linear production functions, $k_{i,1}$ may be zero for some agent $i$ (see Proposition \ref{linear2-result} below and Remark \ref{continuum} in Appendix  \ref{appendix}).} Then, thanks to the assumption $u_i^{\prime}(c)+cu_i^{\prime \prime}(c)\geq 0$, $\forall c>0$, we can prove that $b_{i,0}$ is strictly increasing in $R_1$. So, the market clearing condition $\sum_{i}b_{i,0}=0$ implies that $R_1$ is uniquely determined. By consequence, the equilibrium is unique and is continuous in $(T_0, T_1, D_0)$. 

\paragraph{Step 2.} Our second step is to prove the existence of a political-economic equilibrium with transfers and public investment in Definition \ref{def-equilibrium-2}. The following claim is immediate from the definition of political-economic equilibrium.
\begin{lemma}\label{selectionlemma0}
There exists a political-economic equilibrium  if and only if there exist $T_1^*\geq 0$ and a competitive equilibrium $\big((c_{i,0}^*,c_{i,1}^*,k_{i,1}^*, b_{i,0}^*)_{i=1}^m, R_1^*\big)$ associated to the tax $T_1^*$ such that
\begin{align}
T_1^*=\sum_{i=1}^m\tau_{i,1}\mathcal{A}_i(D_0)F_i(k_{i,1}^*).
\end{align}
\end{lemma}
\begin{proof}
Define $T_0^*\equiv \sum_{i=1}^m\tau_{i,0} w_{i,0}$. Then, take $T_1^*$  and  an equilibrium $\Big((c_{i,0}^*,c_{i,1}^*,k_{i,1}^*, b_{i,0}^*)_{i=1}^m, R_1^*\Big)$ associated to the tax $T_1^*$,  as in the statement of Lemma \ref{selectionlemma0}. Then, we have a political-economic equilibrium. 
\end{proof}

It means that $T_1^*$ in Lemma \ref{selectionlemma0} is a fixed point. To prove the existence of such a fixed point, we require Assumption \ref{assumption-additional}  and make use of the Brouwer fixed point theorem. 
\begin{proposition}\label{prop22}Let Assumptions \ref{assum_Ai}, \ref{assum-agenti}, and \ref{assumption-additional}  be satisfied.   Then, there exists a political-economic equilibrium with transfers and innovation.
\end{proposition}
\begin{proof}
For each $(T_0,T_1)\in [0,\bar{W}]^2$ and $D_0=\gamma_{d,0}T_0$ where $\gamma_d$ is a parameter in $[0,1]$, there exists a unique  competitive equilibrium (as in Definition \ref{def-equilibrium-0}) whose outcomes are continuous in $(T_0,T_1)$. Denote the capital allocation by $k_{i,1}(T_0,T_1)$ for $i=1,\ldots,m$. 

Then, we define the map $\Gamma: [0,\bar{W}]^2\to [0,\bar{W}]^2$ by the following:
\begin{align}
\Gamma_1(T_0,T_1)&=\sum_{i=1}^m\tau_{i,0} w_{i,0}, \quad 
\Gamma_2(T_0,T_1)=\sum_{i=1}^m\tau_{i,1}\mathcal{A}_i(D_0)F_i(k_{i,1}(T_0,T_1))
\end{align}

The function $\Gamma_1$ is well-defined and constant because both lists $(\tau_{i,0})$ and $(w_{i,0})$ are exogenous (recall that Definition \ref{def-equilibrium-2}'s point \ref{defpoint2} requires that $T_0=\sum_{i=1}^m\tau_{i,0} w_{i,0}$). Thanks to Assumption \ref{assumption-additional} and Proposition  \ref{prop12}, there exists a unique competitive equilibrium. By consequence,  the function $\Gamma_2$ is well-defined.

 By Proposition  \ref{prop12}, the map $\Gamma$ is continuous. Moreover, by Remark \ref{remark2}, we have $\Gamma_1(T_0,T_1)\in [0,\bar{W}]$ and $\Gamma_2(T_0,T_1)\in [0,\bar{W}]$.  Applying the Brouwer fixed point theorem, there exists $(T_0^*,T_1^*)\in [0,\bar{W}]^2$ satisfying 
\begin{align*}
T_0^*&=\Gamma_1(T_0^*,T_1^*)=\sum_{i=1}^m\tau_{i,0} w_{i,0}, \quad 
T_1^*=\Gamma_2(T_0^*,T_1^*)=\sum_{i=1}^m\tau_{i,1}\mathcal{A}_i(D_0)F_i(k_{i,1}(T_0^*,T_1^*)).
\end{align*}
It means that this couple $(T_0^*,T_1^*)$ together with its associated competitive equilibrium (as in Definition \ref{def-equilibrium-0}) constitutes a political-economic equilibrium following Definition \ref{def-equilibrium-2}. We have finished our proof.
\end{proof}

\paragraph{Discussion of Assumption \ref{assumption-additional}.}

We require Assumption \ref{assumption-additional}  to ensure that the map $\Gamma$ is well defined and continuous, which allows us to make use of the Brouwer fixed-point theorem. Without Assumption  \ref{assumption-additional}, there may exist multiple competitive equilibria.\footnote{See Remark \ref{continuum} in Appendix \ref{appendix}.} In this case, we can define a correspondence  $\Gamma: [0,\bar{W}]^2	\rightrightarrows [0,\bar{W}]^2$ by the following.
\begin{subequations}
\begin{align}
\Gamma_1(T_0,T_1)&=\{T_0^*\}, \text{ where }T_0^*\equiv \sum_{i=1}^m\tau_{i,0} w_{i,0}\\ 
\Gamma_2(T_0,T_1)&=\Big\{\sum_{i=1}^m\tau_{i,1}\mathcal{A}_i(D_0)F_i(k_{i,1}(T_0,T_1)): \notag\\
&\quad \quad \quad (k_{i,1}(T_0,T_1))_{i=1}^m \text{ is an equilibrium allocation of capital}\Big\}.
\end{align}
\end{subequations}
There exists a political-economic equilibrium if and only if there exists a so-called fixed point $T^*_1$ such that $T^*_1 \in \Gamma_2(T_0^*,T_1^*)$. This is satisfied if there exists a continuous selector of the correspondence $\Gamma$. So, it remains to prove the existence of a continuous selector. We may do so by applying the Michael selection theorem \citep{Michael1956}\footnote{See also \cite{ab06infinite}, Section 17.11.} for the correspondence $\Gamma$, but additional assumptions are needed because the Michael selection theorem requires that the correspondence is lower hemicontinuous and has nonempty closed convex values. When there exist multiple competitive equilibria,  the correspondence $\Gamma$ may not have convex values.

 Theorem 1 in \cite{mitra98} and  Theorem 1 in \cite{levan02} establish the existence of an equilibrium with externality. However, since \cite{mitra98}, 
\cite{levan02} consider models with one representative consumer, they do not need the assumption $u_i^{\prime}(c)+cu_i^{\prime \prime}(c)\geq 0$ $\forall c>0$.

\section{Analytical results}
\label{analyses}

For the sake of tractability, we consider a simple model with two groups of agents ($m = 2$) to obtain explicit analytical results.

\begin{assum}\label{tractable1}Assume that there are two groups $i=1,2$ with logarithmic utility $u_i(c)=ln(c)$, agents in each group are identical so that we have two representative agents, and \begin{align}
F_i(k)&=k, \quad A_i(D)=A_i(1+a_iD) \text{ for any } k, D.
\end{align}
where $A_i$ is the autonomous productivity and $a_i$ represents the effect of the R\&D investment on agent i's productivity.
\end{assum}


\subsection{Equilibrium outcomes with and without intervention}
First, we compute the equilibrium in the absence of intervention.
\begin{lemma}\label{without-intervention}Let Assumption \ref{tractable1} be satisfied. Assume that two productivities are different, i.e., $A_2\not=A_1$. Focus on the case without interventions: $\tau_{i,t}=\gamma_{i,t}=0$ for any $i,t$. There exists a unique equilibrium. In equilibrium, we have
\begin{subequations}\label{outcome0}
\begin{align}\label{K10}K_1^*&=k_{1,1}+k_{2,1}=\frac{\beta_2}{1+\beta_2}w_{2,0}+\frac{\beta_1}{1+\beta_1}w_{1,0}\\
\label{Y10}Y_1^*&=\Big(\frac{\beta_2}{1+\beta_2}w_{2,0}+\frac{\beta_1}{1+\beta_1}w_{1,0}\Big)\max(A_1,A_2).
\end{align}\end{subequations}
\end{lemma}
\begin{proof}
See Appendix \ref{appendix}.
\end{proof}
In the absence of government intervention, the economic outcomes (aggregate capital $K_1$ and aggregate GDP $Y_1$) at date $1$ are defined by the initial endowment $\omega_{i,0}$, saving/investment preferences $\beta_i$ (i.e., rate of time discount), and autonomous productivities $A_i$, $i = 1,2$


In the presence of intervention, we have the following result.

\begin{proposition}\label{linear2-result}
Let Assumption \ref{tractable1} be satisfied. Assume that the productivity of agent 2 at date $1$ after government intervention  is higher than that of agent 1, i.e., 
\begin{align}\label{linear2assumption0}
(1-\tau_{2,1})(1+a_2\gamma_{d,0}T_0)A_2>(1-\tau_{1,1})(1+a_1\gamma_{d,0}T_0)A_1,
\end{align}
where $T_0=\tau_{1,0} w_{1,0}+\tau_{2,0} w_{2,0}.$  Then, there exists a unique equilibrium.\footnote{When $(1-\tau_{2,1})(1+a_2\gamma_{d,0}T_0)A_2<(1-\tau_{1,1})(1+a_1\gamma_{d,0}T_0)A_1$, we have $k_{1,1}>0$ and $k_{2,1}=0$ in equilibrium. When $(1-\tau_{2,1})(1+a_2\gamma_{d,0}T_0)A_2=(1-\tau_{1,1})(1+a_1\gamma_{d,0}T_0)A_1$, there may exist a continuum of equilibria (see Remark \ref{continuum} in Appendix  \ref{appendix}).} In equilibrium, we have $k_{1,1}=0$ and $k_{2,1}>0$.
\begin{enumerate}
\item The aggregate capital is equal to
\begin{align}\label{K1}
K_1&=k_{2,1}=\dfrac{\frac{\beta_2\big[(1-\tau_{2,0})w_{2,0}+\gamma_{2,0}(\tau_{1,0} w_{1,0}+\tau_{2,0} w_{2,0})\big]}{(1+\beta_2)}+\frac{\beta_1\big[(1-\tau_{1,0})w_{1,0}+\gamma_{1,0}(\tau_{1,0} w_{1,0}+\tau_{2,0} w_{2,0})\big]}{(1+\beta_1)}}{1+(\frac{\gamma_{2,1}}{1+\beta_2}+\frac{\gamma_{1,1}}{1+\beta_1})\frac{\tau_{2,1}}{(1-\tau_{2,1})}}
\end{align} 
$K_1$ is increasing in $\gamma_{1,0}$ and $\gamma_{2,0}$ but decreasing in $\gamma_{1,1}$ and $\gamma_{2,1}$.

\item The GDP of the economy at date $1$ equals 
\begin{align}
Y_1&=\big(1-\tau_{2,1}\gamma_{c,1}\big)\big(1+a_2(1-\gamma_{1,0}-\gamma_{2,0}-\gamma_{c,0})T_0\big)A_2k_{2,1}.
\end{align} 
\end{enumerate}

\end{proposition}
\begin{proof}
See Appendix \ref{appendix}.
\end{proof}

Condition (\ref{linear2assumption0}) means that the after-tax productivity of agent $2$ is higher than that of agent 1. Since we are considering linear production functions, condition (\ref{linear2assumption0}) guarantees that $k_{1,1}=0$ and $k_{2,1}>0$, i.e., the total capital $K_1$ is used by the agent $2$ (the most productive agent) and the agent $1$ does not produce.

Some remarks deserve to be mentioned.
\begin{enumerate}
\item {\bf R\&D efficiency}. GDP at date $1$, $Y_1$, is increasing in the efficiency $a_2$ of R\&D. The higher the R\&D process's efficiency, the higher GDP.
\item {\bf Effect of corruption}. GDP at date $1$ is naturally decreasing with the fraction of resource lost  $\gamma_{c,1}$ (i.e., $1-\gamma_{1,1}-\gamma_{2,1}$) and $\gamma_{c,0} $ (recall that $\gamma_{c,0} =1-\gamma_{1,0}-\gamma_{2,0} - \gamma_{d,0}$). It means that the corruption is always harmful for the economic development.
\end{enumerate}



\subsection{Effects of government interventions}
To investigate the effects of government interventions, we compare these outcomes with those in the case without interventions and provide the comparative statics regarding the role of distribution policy $(\tau_{i,t},\gamma_{i,t})_{i=1}^2$ and government effort $\gamma_{d,0}$ in R\&D.

\begin{proposition}\label{corollary1}
Let Assumptions in Proposition \ref{linear2-result} hold. In equilibrium, the gap between the GDP of the economy with interventions and without interventions is 
\begin{align}Y_1-Y_1^*
=&(1-\tau_{2,1}\gamma_{c,1})(1+a_2\gamma_{d,0}T_0)A_2
\dfrac{\frac{\beta_2\big[(1-\tau_{2,0})w_{2,0}+\gamma_{2,0}T_0\big]}{(1+\beta_2)}+\frac{\beta_1\big[(1-\tau_{1,0})w_{1,0}+\gamma_{1,0}T_0\big]}{(1+\beta_1)}}{1+(\frac{\gamma_{2,1}}{1+\beta_2}+\frac{\gamma_{1,1}}{1+\beta_1})\frac{\tau_{2,1}}{(1-\tau_{2,1})}} \notag\\
\label{outputdifference}&\quad -\Big(\frac{\beta_2}{1+\beta_2}w_{2,0}+\frac{\beta_1}{1+\beta_1}w_{1,0}\Big)\max(A_1,A_2),
\end{align}
where $T_0=\tau_{1,0} w_{1,0}+\tau_{2,0} w_{2,0}$ and $\gamma_{c,1}\equiv 1-\gamma_{1,1}-\gamma_{2,1}.$

The capital gap is \begin{align*}
K_1-K_1^*&=\dfrac{\frac{\beta_2\big[(1-\tau_{2,0})w_{2,0}+\gamma_{2,0}(\tau_{1,0} w_{1,0}+\tau_{2,0} w_{2,0})\big]}{(1+\beta_2)}+\frac{\beta_1\big[(1-\tau_{1,0})w_{1,0}+\gamma_{1,0}(\tau_{1,0} w_{1,0}+\tau_{2,0} w_{2,0})\big]}{(1+\beta_1)}}{1+(\frac{\gamma_{2,1}}{1+\beta_2}+\frac{\gamma_{1,1}}{1+\beta_1})\frac{\tau_{2,1}}{(1-\tau_{2,1})}}
\\
& - \Big(\frac{\beta_2}{1+\beta_2}w_{2,0}+\frac{\beta_1}{1+\beta_1}w_{1,0}\Big).
\end{align*} 

\end{proposition}

Observe that $Y_1-Y_1^*$ is continuously increasing in the R\&D efficiency $a_2$. Moreover, if $\gamma_{d,0}T_0>0$, we have $\lim_{a_2\to\infty}(Y_1-Y_1^*)=\infty$. By consequence, we obtain the following result.
\begin{proposition}[inaction versus intervention and corruption: role of efficiency]\label{inaction} Let Assumptions in Proposition \ref{linear2-result} hold. 
Assume that the government effort in R\&D is significant, i.e., $\gamma_{d,0}T_0>0$. 
\begin{enumerate}
\item \label{proposition5_2} 
When the efficiency $a_2$ is high enough, we have $Y_1-Y_1^*>0$. This happens even if there is a corruption, i.e.,  $1>\gamma_{1,0}+\gamma_{2,0}+\gamma_{d,0}$ and $ 1>\gamma_{1,1}+\gamma_{2,1}$. 

\item \label{proposition5_2} When the efficiency  $a_2$ is low and/or the corruption is high (i.e., $1-\gamma_{1,1}-\gamma_{2,1}$ is high), then $Y_1-Y_1^*<0$.
\end{enumerate}
\end{proposition}

As mentioned above, the presence of corruption is always harmful to economic development. However,  Proposition \ref{inaction} argues that the inaction may be worse. Indeed, Proposition \ref{inaction}'s point 1 shows that when the government provides a significant investment in public investment and it is quite efficient (i.e. high value of $a_2$), we obtain more economic output with respect to the inaction situation, even if a part of tax revenue is wasted due to corruption.

Proposition \ref{inaction} is related to \cite{dl02}'s Proposition 3 where they show, in an optimal growth framework, that a scenario, where embezzling improves the productivity,  corruption is not very high and the incentive effect is important, would be better than the inaction scenario (where the productivity remains the same and there is no corruption).

\begin{Simulation}\label{simulation1}{\normalfont We illustrate the above theoretical findings by a numerical simulation. We set the following baseline parameters:

\begin{center}    
\begin{tabular}{|l|l|c|l|}
\hline
\textbf{Description} & \textbf{Parameter}   \\
\hline
 Initial endowments & $w_{1,0}=10, w_{2,0}=1$ \\
 Autonomous productivity & $A_1=1, A_2=1.1$   \\
Time discount rate & $\beta_1=\beta_2=0.6$ \\
Tax rates at periods $0$ and $1$ & $\tau_{1,0}=\tau_{2,0}=0.1$ \text{ and } $\tau_{1,1}=\tau_{2,1}=0.1$   \\
 Transfers at periods $t=0$ and  $t=1$  & $\gamma_{1,0}= \gamma_{2,0}=0.41 $ and $\gamma_{1,1}=\gamma_{2,1}=0.47$ \\
  R\&D efficiency & $ a_1=0.85, a_2=0.85$  \\
Public investment share (R\&D) & $\gamma_{d,0}=0.1$  \\
Corruption rate at date $0$  & $\gamma_{c,0}=1-\gamma_{d,0}-\sum_i\gamma_{i,0}$ = 0.08 \\
Corruption rate  at date $1$  & $\gamma_{c,1}=1-\sum_i\gamma_{i,1}$ =0.06. \\
\hline
\end{tabular}
\end{center}

We take $\beta_1=\beta_2=0.6$ so that the saving rate in our two-period model is $\frac{\beta_1}{1+\beta_1}=37.5\%$. According to the database of the World Bank,\footnote{See https://data.worldbank.org/.} the gross savings (\% of GDP) of  China and Viet Nam in 2023 is 42\% and 35\% respectively.

We choose the tax rates $\tau_{1,0}=\tau_{2,0}=\tau_{1,1}=\tau_{2,1}=10\%$ so that the tax revenue to the GDP equals $\frac{T_0}{w_{1,0}+w_{2,0}}=10\%$. According to the database of the World Bank, the tax revenue (\% of GDP) of Brazil, China, France in 2023 is 14\%, 7.6\%, 23.1\%, respectively.

We choose $w_{1,0}=10, w_{2,0}=1$ to capture the fact that the size of the type-1 agents is 10 times bigger than that of the type-2 agents. Conditions $A_2=1.1>A_1=1$ means that the type-2 agent is more productive than the type-1 agent.

Our choice $\gamma_{d,0}=0.1$ means that public investment in R\&D represents $10\%$ of total tax revenue. 

With these specifications, the condition 
(\ref{linear2assumption0}) in Proposition \ref{linear2-result} holds (indicating that the productivity of agent 2 after government intervention is higher than that of agent 1). According to (\ref{outputdifference}) in Proposition \ref{corollary1},  we find that
\begin{align}
Y_1-Y_1^* \approx 0.01 >0 \text{ and } \frac{Y_1-Y_1^*}{Y_1^*}\approx 0.2\%>0.
\end{align}
This means that the inaction is worse than the situation where public investment in R\&D is highly efficient (corresponding to a high $a_2$) despite there is corruption. Note that the corruption rate at date $1$ is $\gamma_{c,1}=0.06$.

We now set $\gamma_{1,1}=\gamma_{2,1}=0.42$, which implies that the corruption rate at date $1$ is $\gamma_{c,1}=1-\gamma_{1,1}-\gamma_{2,1}=0.16$, which is higher than the previous level $0.06$.  In this case, we see that $\frac{Y_1-Y_1^*}{Y_1^*}\approx -0.16\%<0$. This means that the output decreases if the corruption level is high (this point is consistent with Proposition \ref{inaction}'s point \ref{proposition5_2}).
}
\end{Simulation}

We now focus on the case without public investment in R\&D, i.e., $\gamma_{d,0}=0$. Could the output in the economy with intervention and corruption be higher than the output in the absence of intervention? We argue this may be the case. The following result focuses on the role of the time discount rate.

\begin{proposition}[inaction versus intervention and corruption: role of saving rate]\label{inaction2} Let Assumption \ref{tractable1} be satisfied. Assume also that $\gamma_{d,0}=0$, $A_2>A_1$, and there is no  intervention at date $1$, i.e., $\tau_{1,1}=\tau_{2,1}=0$.  
The output gap equals 
\begin{align}Y_1-Y_1^*
=&A_2\frac{\beta_2}{1+\beta_2}\Big(\gamma_{2,0}\tau_{1,0} w_{1,0}-(1-\gamma_{2,0})\tau_{2,0}w_{2,0}\Big)\notag\\
&+A_2\frac{\beta_1}{1+\beta_1}\Big(\gamma_{1,0}\tau_{2,0} w_{2,0}-(1-\gamma_{1,0})\tau_{1,0} w_{1,0}\Big).
\end{align}
\begin{enumerate}
\item 
If $\beta_1=\beta_2=\beta$, then we can compute that
\begin{align}Y_1-Y_1^*
=&\frac{\beta_1}{1+\beta}A_2(1-\gamma_{1,0}-\gamma_{2,0})T_0=-\frac{\beta_1}{1+\beta}A_2\gamma_{c,0}T_0\leq 0
\end{align}
This is strictly negative if and only if the wasted rate $\gamma_{c,0}>0$. 

\item Assume that $\gamma_{2,0}\tau_{1,0} w_{1,0}-(1-\gamma_{2,0})\tau_{2,0}w_{2,0}>0$. We have $Y_1-Y_1^*>0$  
if $\beta_2$ is high enough. 

\item Assume that $\gamma_{1,0}\tau_{2,0} w_{2,0}-(1-\gamma_{1,0})\tau_{1,0} w_{1,0}<0$. We have $Y_1-Y_1^*<0$ 
if $\beta_1$ is high enough.
\end{enumerate}
\end{proposition}
\begin{proof}See Appendix \ref{appendix}.\end{proof}
When both agents have the same rate of time discount, the presence of corruption lowers the aggregate output. 

We now explain the intuition behind point 2 of Proposition \ref{inaction2}. Condition $\gamma_{2,0}\tau_{1,0} w_{1,0}-(1-\gamma_{2,0})\tau_{2,0}w_{2,0}>0$ means that the government's redistribution policy does increase the after-transfer income of agent $2$ (the most productive agent), corresponding to $(1-\tau_{2,0})w_{2,0}+\gamma_{2,0}T_0>w_{2,0}$ as shown in Appendix \ref{appendix}. By combining with the fact that the most productive agent has a higher rate of time discount and a higher saving rate, the aggregate capital would be higher than that in the absence of intervention, even if there is some corruption $\gamma_{c,0}>0$. So, the economy's output would be higher.

Point 3 of Proposition \ref{inaction2} shows another story. Condition $\gamma_{1,0}\tau_{2,0} w_{2,0}-(1-\gamma_{1,0})\tau_{1,0} w_{1,0}<0$ corresponds to $(1-\tau_{1,0})w_{1,0}+\gamma_{1,0}T_0 <w_{1,0}$ as shown in Appendix \ref{appendix}: a distorted redistribution may be harmful for economic growth, even if there is no corruption (i.e., $\gamma_{1,0}+\gamma_{2,0} =1$). This happens if the agent 1 (whose rate of time discount $\beta_1$ and saving rate are high) is taxed too much, because this reduces the savings of this agent, which in turn decreases the aggregate investment. 


We now emphasize another bad intervention.
\begin{proposition}[harmful taxation]\label{bad3}Let Assumption \ref{tractable1} be satisfied. Assume that $\beta_1=\beta_2=\beta>0$ and we abstract from the public investment (i.e., $\gamma_{d,0}=0$). 

Assume that $A_1>A_2>\frac{1-\tau_{1,1}}{1-\tau_{2,1}}A_1>0$. Then, we have $Y_1<Y_1^*$ in equilibrium.
\end{proposition}
\begin{proof}See Appendix \ref{appendix}.\end{proof}

In this economy, at the beginning, the productivity of agent $1$ is higher than that of agent $2$. However, when the government sets a high tax rate $\tau_{1,1}$ on the most productive agent (agent $1$ in this case) in the sense that $A_2>\frac{1-\tau_{1,1}}{1-\tau_{2,1}}A_1>0$, the after-tax productivity of agent $2$ is higher than that of agent $1$, i.e., $(1-\tau_{2,1})A_2>(1-\tau_{1,1})A_1$. This implies that agent $2$ produces while agent $1$ does not in equilibrium. In this case, the new equilibrium output $Y_1$ is lower than $Y^*_1$, which is the output in the case without interventions (see Lemma \ref{without-intervention}). It should be noticed that this happens whatever there is corruption (i.e., $\gamma_{1,t}+\gamma_{2,t}=1$ for $t=0,1$) or not.

\begin{Simulation}\label{simulation2}{\normalfont We illustrate Proposition \ref{bad3} by considering an economy with the same parameter values as in Simulation \ref{simulation1}, except $A_1 = 1.3 > A_2=1.2$, $\tau_{1,1} = 0.2$ and $\gamma_{d,0} = 0$. Condition $A_1>A_2>\frac{1-\tau_{1,1}}{1-\tau_{2,1}}A_1>0$ is satisfied as $1.3>1.2>1.16\approx \frac{1-\tau_{1,1}}{1-\tau_{2,1}}A_1$. 
The gap between $Y_1$ and $Y_1^*$ becomes negative in this case ($\frac{Y_1-Y_1^*}{Y^*_1}=-15.7\%<0$).
}
\end{Simulation}
\section{Conclusion}
\label{conclusion}

We have developed a general equilibrium model that incorporates public investment and corruption. The existence of equilibrium is established using a two-step fixed-point argument. We have then analyzed the role of various redistribution and public investment policies, as well as the effect of corruption. Our findings indicate that corruption consistently is detrimental to economic development. However, policy inaction, defined as the absence of both redistribution and public investment, may lead to worse outcomes than scenarios in which corruption coexists with innovation-led policies.

The results demonstrate that under certain conditions, such as high efficiency of R\&D investment or well-targeted redistribution strategies, the gains from public intervention can compensate for the losses due to corruption. This confirms that governance quality and, particularly, the productivity of public resources always play an essential role in determining the effect of policy interventions. In particular, targeting transfers toward more productive agents or investing in high-impact R\&D sector can mitigate the negative effects of corruption and even improve aggregate output. Furthermore, our analysis reveals that distorted taxation policies, such as over-taxing highly productive or high-saving agents, can be particularly damaging, even in the absence of corruption. 

Overall, this study contributes to a better understanding of the conditions under which government intervention can be justified, even in environments where corruption is prevalent. These findings warn against adopting anti-corruption strategies that neglect the broader context of fiscal policy design. Instead, they call for an integrated approach that considers the interaction between corruption, public investment efficiency, redistribution, and taxation policies.

\appendix
\section{Formal proofs}\label{appendix}

\begin{proof}[{\bf Proof of Proposition \ref{prop12}}]

{\bf Step 1.}  By using the standard approach (see \cite{lvp16} for instance), we can prove that there exists a competitive equilibrium.\\
\newline
{\bf Step 2.}  We prove the following lemma.
\begin{lemma}
\label{implicit1}
For given $R_1>0,T_1\geq 0$, there exists a unique solution $(c_{i,0},c_{i,1}, k_{i,1}, b_{i,0})$ to the  maximization problem (\ref{agentiproblem}) of agent $i$. Moreover, $c_{i,0},c_{i,1}, k_{i,1}$ and $ b_{i,0}$ are continuously differentiable in $(R_1,T_1)$.   
\end{lemma}
\begin{proof}[Proof of Lemma \ref{implicit1}]
Given $T_1\geq 0$ and  $R_1>0$, consider the maximization problem (\ref{agentiproblem}) of agent $i$. It is clear that there exists a solution $(c_{i,0},c_{i,1}, k_{i,1}, b_{i,0})$. Since the utility function is strictly concave, there exists a unique consumption allocation $(c_{i,0},c_{i,1})$ of the agent $i$, and so does $k_{i,1}- b_{i,0} $. We claim that there exists a unique $b_{i,0}$. 

By Inada's condition $F_i'(0)=\infty$, we have  $k_{i,1}>0$. We can write standard first-order conditions (FOC):
\begin{align}
\label{foc1}
u_i'(c_{i,0})&=\beta_i (1-\tau_{i,1})A_iF_i'(k_{i,1})u_i'(c_{i,1})\\
\label{foc2}u_i'(c_{i,0})&=\beta_i R_1 u_i'(c_{i,1}).
\end{align}
So, $R_1=(1-\tau_{i,1})A_iF_i'(k_{i,1})$. When $F_i$ is strictly concave, this implies that $k_{i,1}$ is  continuously differentiable and decreasing in $R_1$.

Denote $x_{i,0}\equiv (1-\tau_{i,0})w_{i,0}+\gamma_{i,0}T_0$. The FOC (\ref{foc2}) gives
\begin{align}u_i'(x_{i,0}+ b_{i,0} -k_{i,1})&=\beta_i R_1u_i'\Big((1-\tau_{i,1})A_iF_i(k_{i,1})-R_1 b_{i,0} +\gamma_{i,1}T_1\Big). \label{eulerproof}
\end{align}
Since $u_i'$ is strictly decreasing ($u_i$ is strictly concave),  this equation has at most 1 solution. This implies that this equation has a unique solution $b_{i,0}$, which implies the uniqueness of $b_{i,0}$.

We have just proved the existence and uniqueness of the solution $(c_{i,0},c_{i,1}, k_{i,1}, b_{i,0})$. It is easy to see that this solution is continuous in $R_1$ and $T_1$.

Given $R_1$ and $T_1$, we define 
\begin{align}
E\equiv \{(k,b): k\geq 0, x_{i,0}+ b -k\geq 0,  (1-\tau_{i,1})A_iF_i(k)-R_1 b +\gamma_{i,1}T_1\geq 0\}.
\end{align}
and the map 
$\Phi: E\to \rr^2$ by 
\begin{align}
\Phi_1(k,b)&\equiv (1-\tau_{i,1})A_iF_i'(k)-R_1\\
\Phi_2(k,b)&\equiv u_i'(x_{i,0}+ b -k)-\beta_i R_1u_i'\Big((1-\tau_{i,1})A_iF_i(k)-R_1 b+\gamma_{i,1}T_1\Big). 
\end{align}
Of course, $\Phi$ is in $C^1$.

To apply the implicit function theorem, we compute
\begin{subequations}
\begin{align}
\frac{\partial \Phi_1}{\partial k}&=(1-\tau_{i,1})A_iF_i^{\prime \prime}(k),\quad  \frac{\partial \Phi_1}{\partial b}=0\\
\frac{\partial \Phi_2}{\partial k}&=-u_i^{\prime \prime}(x_{i,0}+ b -k)-(1-\tau_{i,1})A_iF_i'(k)\beta_i R_1u_i^{\prime \prime}\Big((1-\tau_{i,1})A_iF_i(k)-R_1 b+\gamma_{i,1}T_1\Big) \notag\\
\frac{\partial \Phi_2}{\partial b}&=u_i^{\prime \prime}(x_{i,0}+ b -k)+\beta_i R_1^2u_i^{\prime \prime}\Big((1-\tau_{i,1})A_iF_i(k)-R_1 b+\gamma_{i,1}T_1\Big).
\end{align}\end{subequations}

Since $F_i^{\prime \prime}<0$ and $u_i^{\prime \prime}<0$, the matrix $D\Phi(k,d)\equiv \Big(\frac{\partial \Phi}{\partial k},\frac{\partial \Phi}{\partial b}\Big)^T$ valued at the point $(k_{i,1},b_{i,0})$  is invertible. So, by the implicit function theorem, $k_{i,1}$ and $b_{i,0}$ are continuously differentiable in $(R_1,T_1)$. So are $c_{i,0}$ and $c_{i,1}$.
\end{proof}
{\bf Step 3.} Let us prove the uniqueness of equilibrium allocation and interest rate $R_1$. 

 Taking the derivative of both sides of (\ref{eulerproof}) with respect to $R_1$ we have 
\begin{align*}
u_i^{\prime \prime}(c_{i,0})\big( b_{i,0} '(R_1)-k_{i,1}'(R_1)\big)&=\beta_iu_i^{\prime}(c_{i,1})\\
&+\beta_1R_1u_i^{\prime \prime}(c_{i,0})\Big((1-\tau_{i,1})A_iF_i'(k_{i,1})k_{i,1}'(R_1)- b_{i,0} -R_1 b_{i,0} '(R_1)\Big).
\end{align*}
Combining with $(1-\tau_{i,1})A_iF_i'(k_{i,1})=R_1$, we get that
\begin{align}\label{additionalstep}
&u_i^{\prime \prime}(c_{i,0})\big( b_{i,0} '(R_1)-k_{i,1}'(R_1)\big)=\beta_iu_i^{\prime}(c_{i,1})+\beta_1R_1u_i^{\prime \prime}(c_{i,1})\Big(R_1k_{i,1}'(R_1)- b_{i,0} -R_1 b_{i,0} '(R_1)\Big).
\end{align}
This implies that
\begin{align*}
&\Big(u_i^{\prime \prime}(c_{i,0})+\beta_1R_1^2u_i^{\prime \prime}(c_{i,1})\Big)\big( b_{i,0} '(R_1)-k_{i,1}'(R_1)\big)=\beta_iu_i^{\prime}(c_{i,1})-\beta_iR_1 b_{i,0} u_i^{\prime \prime}(c_{i,1}).
\end{align*}
According to the budget constraint $-R_1 b_{i,0}=c_{i,1}-(1-\tau_{i,1})A_iF_i(k_{i,1})- \gamma_{i,1}T_1$, we have
\begin{align*}
&\beta_iu_i^{\prime}(c_{i,1})-\beta_iR_1 b_{i,0} u_i^{\prime \prime}(c_{i,1})\\
&=\beta_i\Big(u_i^{\prime}(c_{i,1})+c_{i,1}u_i^{\prime \prime}(c_{i,1})\Big)+\beta_iu_i^{\prime \prime}(c_{i,1})\Big(-(1-\tau_{i,1})A_iF_i(k_{i,1})-\gamma_{i,1}T_1\Big)>0
\end{align*}
because $u^{\prime}(c)+cu^{\prime \prime}(c)\geq 0$ for any $c>0$. It means that $\beta_iu_i^{\prime}(c_{i,1})-\beta_iR_1 b_{i,0} u_i^{\prime \prime}(c_{i,1})>0$. Combining with (\ref{additionalstep}), we have $ b_{i,0} '(R_1)-k_{i,1}'(R_1)<0$. Since $k_{i,1}'(R_1)<0$. We obtain $ b_{i,0} '(R_1)<0$.

So, $ b_{i,0} $ is strictly decreasing in $R_1$. Combining with the market clearing condition $\sum_{i=1}^m b_{i,0} =0$, we get that $R_1$ is uniquely determined. \\
\newline
{\bf Step 4.} Since  $R_t$ is uniquely determined by $\sum_{i=1}^m b_{i,0} =0$, where $b_{i,0}$ are continuously differentiable in $(R_1,T_0,T_1)$, the equilibrium interest rate $R_1$ is continuously differentiable in $(T_0,T_1)$.

\end{proof}

\begin{proof}[{\bf Proof of Lemma \ref{without-intervention}}]
 Agent $i$'s problem becomes
\begin{subequations}
\begin{align*}
&\max_{(c_{i,0},c_{i,1}, k_{i,1}, b_{i,0} )}u_i(c_{i,0})+\beta_iu(c_{i,1})\\
\text{subject to constraints: } &c_{i,0}+k_{i,1}\leq w_{i,0}+ b_{i,0} \\
&c_{i,1}\leq A_ik_{i,1}-R_1 b_{i,0} \\
&c_{i,0}\geq 0, c_{i,1}\geq 0, k_{i,1}\geq 0
\end{align*}
\end{subequations}

Without loss of generality, assume that $A_2>A_1$. Since there is no borrowing constraint, the agent $1$ does not produce and we find that
\begin{align*}
R_1&=A_2\\
- b_{1,0} &=\frac{\beta_1}{1+\beta_1}w_{1,0}\\
k_{2,1}- b_{2,0} &=\frac{\beta_2}{1+\beta_2}w_{2,0}\\
\Rightarrow k_{2,1}&=\frac{\beta_2}{1+\beta_2}w_{2,0}+ b_{2,0} =\frac{\beta_2}{1+\beta_2}w_{2,0}- b_{1,0} =\frac{\beta_2}{1+\beta_2}w_{2,0}+\frac{\beta_1}{1+\beta_1}w_{1,0}.
\end{align*}
The output at date $1$ equals:
\begin{align*}
Y_1=A_2k_{2,1}=A_2\Big(\frac{\beta_2}{1+\beta_2}w_{2,0}+\frac{\beta_1}{1+\beta_1}w_{1,0}\Big).
\end{align*}
\end{proof}

\begin{proof}[{\bf Proof of Proposition \ref{linear2-result}}]
Consider an equilibrium. We have $T_0=\sum_{i=1}^m\tau_{i,0} w_{i,0}$. So, we have $D_0=\gamma_{d,0}T_0=\gamma_d\sum_{i=1}^2\tau_{i,0} w_{i,0}.$
The budget constraints of agent $i$ are
\begin{align*}
 &c_{i,0}+k_{i,1}\leq (1-\tau_{i,0})w_{i,0}+ b_{i,0} +\gamma_{i,0}T_0\\
&c_{i,1}\leq (1-\tau_{i,1})\mathcal{A}_i(D_0)k_{i,1}-R_1 b_{i,0} +\gamma_{i,1}T_1.
\end{align*}

We write the FOCs
\begin{align*}u_i'(c_{i,0})&=\beta_i R_1 u_i'(c_{i,1})\\
u_i'(c_{i,0})&=\beta_i (1-\tau_{i,1})\mathcal{A}_i(D_0)F_i'(k_{i,1})u_i'(c_{i,1})+\mu_{i,1}=\beta_i (1-\tau_{i,1})\mathcal{A}_i(D_0)u_i'(c_{i,1})+\mu_{i,1}\\
 \mu_{i,1}&\geq 0,\quad  \mu_{i,1}k_{i,1}=0.
\end{align*}
This implies that $R_1\geq (1-\tau_{i,1})\mathcal{A}_i(D_0)$ for any $i$.

Since $(1-\tau_{2,1})(1+a_2\gamma_{d,0}T_0)A_2>(1-\tau_{1,1})(1+a_1\gamma_{d,0}T_0)A_1$, we have 
\begin{align}\label{R1formulas}
R_1=(1-\tau_{2,1})\mathcal{A}_2(D_0).\end{align}
By consequence,  $R_1=(1-\tau_{2,1})\mathcal{A}_2(D_0)>(1-\tau_{1,1})\mathcal{A}_1(D_0)$, which implies that $\mu_{1,1}>0$. Thus, $k_{1,1}=0$.  

In equilibrium, we have, for $i=1,2$,
\begin{align*}
 &c_{i,0}+k_{i,1}= (1-\tau_{i,0})w_{i,0}+ b_{i,0} +\gamma_{i,0}T_0\\
&c_{i,1}= (1-\tau_{i,1})\mathcal{A}_i(D_0)F_i(k_{i,1})-R_1 b_{i,0} +\gamma_{i,1}T_1.
\end{align*}
Then, we find that 
\begin{align*}
c_{i,1}=\beta_1R_1c_{i,0} \text{ for any } i=1,2\\
-R_1 b_{1,0} +\gamma_{1,1}T_1=\beta_1R_1\big((1-\tau_{1,0})w_{1,0}+ b_{1,0} +\gamma_{1,0}T_0\big)
\end{align*}
Then, we can find the saving of agent $1$ by
\begin{align*}
- b_{1,0} R_1(1+\beta_1)&=\beta_1R_1\Big((1-\tau_{1,0})w_{1,0}+\gamma_{1,0}T_0\Big)-\gamma_{1,1}T_1\\
- b_{1,0} &=\frac{\beta_1R_1\Big((1-\tau_{1,0})w_{1,0}+\gamma_{1,0}T_0\Big)-\gamma_{1,1}T_1}{(1+\beta_1)R_1}\\
&=\frac{\beta_1\Big((1-\tau_{1,0})w_{1,0}+\gamma_{1,0}T_0\Big)}{(1+\beta_1)}-\frac{\gamma_{1,1}T_1}{(1+\beta_1)R_1}.
\end{align*}
We now look at the agent 2's problem:
\begin{align*}
c_{2,1}&=\beta_2R_1c_{2,0}\\
\Leftrightarrow (1-\tau_2)F_2(k_{2,1})-R_1 b_{2,0} +\gamma_{2,1}T_1&=\beta_2R_1\Big((1-\tau_{2,0})w_{2,0}+ b_{2,0} +\gamma_{2,0}T_0-k_{2,1}\Big)\\
\Leftrightarrow (1+\beta_2)R_1(k_{2,1}- b_{2,0} )&=\beta_2R_1\Big((1-\tau_{2,0})w_{2,0}+\gamma_{2,0}T_0\Big)-\gamma_{2,1}T_1.
\end{align*}

Then, combining with $ b_{1,0} + b_{2,0} =0$, we can compute the capital of agent $2$
\begin{align}
k_{2,1}&=\frac{\beta_2R_1\Big((1-\tau_{2,0})w_{2,0}+\gamma_{2,0}T_0\Big)-\gamma_{2,1}T_1}{(1+\beta_2)R_1}+ b_{2,0} \notag\\
\label{k21}&=\frac{\beta_2R_1\Big((1-\tau_{2,0})w_{2,0}+\gamma_{2,0}T_0\Big)-\gamma_{2,1}T_1}{(1+\beta_2)R_1}+\frac{\beta_1R_1\Big((1-\tau_{1,0})w_{1,0}+\gamma_{1,0}T_0\Big)-\gamma_{1,1}T_1}{(1+\beta_1)R_1}.
\end{align}
Now, recall that $$T_1=\tau_{1,1}\mathcal{A}_1(D_0)F_1(k_{1,1})+\tau_{2,1}\mathcal{A}_2(D_0)F_2(k_{2,1})=\tau_{2,1}\mathcal{A}_2(D_0)k_{2,1}.$$
Combining with (\ref{R1formulas}), we have
\begin{align}\label{T1}
\frac{T_1}{R_1}=\frac{\tau_{2,1}\mathcal{A}_2(D_0)k_{2,1}}{(1-\tau_{2,1})\mathcal{A}_2(D_0)}=\frac{\tau_{2,1}}{(1-\tau_{2,1})}k_{2,1}.
\end{align}
By substituting this in (\ref{k21}), we find that
\begin{align*}
&k_{2,1}\Big(1+(\frac{\gamma_{2,1}}{1+\beta_2}+\frac{\gamma_{1,1}}{1+\beta_1})\frac{\tau_{2,1}}{(1-\tau_{2,1})}\Big)=\frac{\beta_2\Big((1-\tau_{2,0})w_{2,0}+\gamma_{2,0}T_0\Big)}{(1+\beta_2)}+\frac{\beta_1\Big((1-\tau_{1,0})w_{1,0}+\gamma_{1,0}T_0\Big)}{(1+\beta_1)}\notag\\
=&\frac{\beta_2\Big((1-\tau_{2,0})w_{2,0}+\gamma_{2,0}(\tau_{1,0} w_{1,0}+\tau_{2,0} w_{2,0})\Big)}{(1+\beta_2)}+\frac{\beta_1\Big((1-\tau_{1,0})w_{1,0}+\gamma_{1,0}(\tau_{1,0} w_{1,0}+\tau_{2,0} w_{2,0})\Big)}{(1+\beta_1)}.
\end{align*}
By consequence, we obtain (\ref{K1}). 

Then, the GDP of the economy at date $1$ equals 
\begin{align*}
Y_1&=\underbrace{\sum_{i=1}^2c_{i,1}}_\text{Consumption}=\sum_{i=1}^2(1-\tau_{i,1})\mathcal{A}_i(D_0)F_{i}(k_{i,1})+\big(\sum_{i=1}^2\gamma_{i,1}\big)T_1\\
&=\sum_{i=1}^2\mathcal{A}_i(D_0)F_{i}(k_{i,1})-\big(1-\sum_{i=1}^2\gamma_{i,1})T_1.
\end{align*}
Recall the notation $\gamma_{c,1}\equiv 1-\sum_{i=1}^2\gamma_{i,1}$.  According to (\ref{T1}), we have $
T_1=\frac{\tau_{2,1}}{1-\tau_{2,1}}R_1k_{2,1}=\tau_{2,1}A_2k_{2,1}$. Then, we have
\begin{align*}
Y_1&=\mathcal{A}_2(D_0)k_{2,1}-\gamma_{c,1}T_1=\big(1-\gamma_{c,1}\tau_{2,1}\big)\mathcal{A}_2(D_0)k_{2,1}\\
\mathcal{A}_2(D_0)&=A_2(1+a_2D_0)=A_2(1+a_2\gamma_{d,0}T_0)=A_2(1+a_2(\gamma_{c,1}-\gamma_{c,0})T_0)
\end{align*}

The consumption of agent $1$ is
\begin{align*}
c_{1,0}&=(1-\tau_{1,0})w_{1,0}+\gamma_{1,0}T_0+ b_{1,0} 
\\
&=(1-\tau_{1,0})w_{1,0}+\gamma_{1,0}T_0-\frac{\beta_1R_1\Big((1-\tau_{1,0})w_{1,0}+\gamma_{1,0}T_0\Big)-\gamma_{1,1}T_1}{(1+\beta_1)R_1}\\
&=\frac{1}{1+\beta_1}\Big((1-\tau_{1,0})w_{1,0}+\gamma_{1,0}T_0\Big)+\frac{\beta_1}{1+\beta_1}\gamma_1\tau_2A_2k_{2,1}\\
c_{1,1}&=\beta_1R_1c_{1,0}=\beta_1(1-\tau_{2,1})A_2c_{1,0}.
\end{align*} 

\end{proof}

\begin{remark}[continuum of equilibria]
\label{continuum}
{\normalfont Let us consider an economy as in Proposition \ref{linear2-result} but we assume that the two productivities (after taxes) are equal, i.e., 
\begin{align}\label{formulaR*}
(1-\tau_{2,1})(1+a_2\gamma_{d,0}T_0)A_2=(1-\tau_{1,1})(1+a_1\gamma_{d,0}T_0)A_1\equiv R^*.    
\end{align} Recall that $T_0=T_0^*\equiv \sum_{i=1}^m\tau_{i,0} w_{i,0}$.

Denote $e_{i,0}\equiv (1-\tau_{1,0})w_{1,0}+\gamma_{1,0}T_0$. 
The set of solutions to the agent $i$'s problem is all lists $(c_{i,0},c_{i,1},k_{i,1}, b_{i,0} )$ satisfying the following condition:
\begin{subequations}\label{ContinuumAgenti}
\begin{align}
k_{i,1}&=b_{i,0}+\frac{\beta_ie_{i,0}}{(1+\beta_i)}-\frac{\gamma_{1,1}T_1}{(1+\beta_i)R^*}>0\\
c_{i,O}&=\frac{e_{i,0}}{(1+\beta_i)}+\frac{\gamma_{1,1}T_1}{(1+\beta_i)R^*}, \quad c_{i,1}=R^*\Big(\frac{\beta_ie_{i,0}}{(1+\beta_i)}-\frac{\gamma_{1,1}T_1}{(1+\beta_i)R^*}\Big)+\gamma_{i,1}T_1
\end{align}
\end{subequations}
Given $T_1\geq 0$, we observe that the set of competitive equilibria is the set of all lists $\big((c_{i,0},c_{i,1},k_{i,1}, b_{i,0} )_{i=1}^m, R_1\big)$ satisfying $R_1=R^*$, $\sum_{i}b_{i,0}=0$ and (\ref{ContinuumAgenti}) for any $i$. 

Assume that $\frac{\beta_ie_{i,0}}{(1+\beta_i)}-\frac{\gamma_{1,1}T_1}{(1+\beta_i)R^*}>0$  for any $i$. In this case, there exist a continuum of equilibrium allocations (however, recall that there exists a unique equilibrium interest rate, which equals $R^*$ in (\ref{formulaR*})).

We now investigate the political-economic equilibrium. To find the fixed point $T_1$, we solve the following equation
\begin{align}
T_1=\sum_{i}\tau_{i,1}(1+a_i\gamma_{d,0}T_0)A_ik_{i,1}
\end{align}
Denote $h_i\equiv \tau_{i,1}(1+a_i\gamma_{d,0}T_0)A_i.$ We look at 
\begin{align}
T_1&=\sum_{i}h_{i,1} \Big(b_{i,0}+\frac{\beta_ie_{i,0}}{(1+\beta_i)}-\frac{\gamma_{i,1}T_1}{(1+\beta_i)R^*}\Big)\\
\Leftrightarrow T_1&=T_1^*\equiv \frac{\sum_{i} h_{i,1} \Big(b_{i,0}+\frac{\beta_ie_{i,0}}{1+\beta_i}\Big)}{1+\sum_{i}h_{i,1}\frac{\gamma_{i,1}}{(1+\beta_i)R^*}}. \label{T1formulas}
\end{align}
To sum up, we obtain the following result:  Assume that $\frac{\beta_ie_{i,0}}{(1+\beta_i)}-\frac{\gamma_{1,1}T_1}{(1+\beta_i)R^*}>0$ $\forall i$. Take $(b_{i,0})_i\in \rr^m$ be such that 
\begin{align}
b_{i,0}+\frac{\beta_ie_{i,0}}{1+\beta_i}-\frac{\gamma_{1,1}T_1}{(1+\beta_i)R^*}>0, \quad b_{i,0}+\frac{\beta_ie_{i,0}}{1+\beta_i}&>0 \text{  }\forall i \text{ and, } \sum_{i}b_{i,0}=0.
\end{align}
Then the list $\big((c_{i,0},c_{i,1},k_{i,1}, b_{i,0} )_{i=1}^m, R_1,T_0,T_1,G_0, G_1, D_0 \big)$ satisfying three conditions (1) $R_1=R_1^*$, (2) $G_0=T_0=T_0^*$, $G_1=T_1=T_1^*$ (see (\ref{T1formulas})) and (3) condition (\ref{ContinuumAgenti}) for any $i$, constitutes a political-economic equilibrium.
}
\end{remark}

\begin{proof}[{\bf Proof of Proposition \ref{inaction2}}]
First, it is easy to check all assumptions in Proposition \ref{corollary1}. Then, by applying Proposition \ref{corollary1}, we have
\begin{align}Y_1-Y_1^*
=&\Big(1-\tau_{2,1}\gamma_{c,1}\Big)A_2
\dfrac{\frac{\beta_2\big[(1-\tau_{2,0})w_{2,0}+\gamma_{2,0}T_0\big]}{(1+\beta_2)}+\frac{\beta_1\big[(1-\tau_{1,0})w_{1,0}+\gamma_{1,0}T_0\big]}{(1+\beta_1)}}{1+(\frac{\gamma_{2,1}}{1+\beta_2}+\frac{\gamma_{1,1}}{1+\beta_1})\frac{\tau_{2,1}}{(1-\tau_{2,1})}}\\
&\quad -\Big(\frac{\beta_2}{1+\beta_2}w_{2,0}+\frac{\beta_1}{1+\beta_1}w_{1,0}\Big)A_2\\
=&A_2\frac{\beta_2}{1+\beta_2}\Big((1-\tau_{2,1}\gamma_{c,1})\frac{(1-\tau_{2,0})w_{2,0}+\gamma_{2,0}T_0}{1+(\frac{\gamma_{2,1}}{1+\beta_2}+\frac{\gamma_{1,1}}{1+\beta_1})\frac{\tau_{2,1}}{(1-\tau_{2,1})}}-w_{2,0}\Big)\\
&+A_2\frac{\beta_1}{1+\beta_1}\Big((1-\tau_{2,1}\gamma_{c,1})\frac{(1-\tau_{1,0})w_{1,0}+\gamma_{1,0}T_0}{1+(\frac{\gamma_{2,1}}{1+\beta_2}+\frac{\gamma_{1,1}}{1+\beta_1})\frac{\tau_{2,1}}{(1-\tau_{2,1})}}-w_{1,0}\Big)
\end{align}
 In this case, the output gap equals 
\begin{align}Y_1-Y_1^*
=&A_2\frac{\beta_2}{1+\beta_2}\Big(\gamma_{2,0}\tau_{1,0} w_{1,0}-(1-\gamma_{2,0})\tau_{2,0}w_{2,0}\Big)\\
&+A_2\frac{\beta_1}{1+\beta_1}\Big((1-\tau_{1,0})w_{1,0}+\gamma_{1,0}T_0-w_{1,0}\Big).
\end{align}
We see that 
\begin{align*}
(1-\tau_{2,0})w_{2,0}+\gamma_{2,0}T_0-w_{2,0}&=\gamma_{2,0}\big(\tau_{1,0} w_{1,0}+\tau_{2,0} w_{2,0}\big)-\tau_{2,0}w_{2,0}\\
&=\gamma_{2,0}\tau_{1,0} w_{1,0}-(1-\gamma_{2,0})\tau_{2,0}w_{2,0}\\
(1-\tau_{1,0})w_{1,0}+\gamma_{1,0}T_0-w_{1,0}&=\gamma_{1,0}\big(\tau_{1,0} w_{1,0}+\tau_{2,0} w_{2,0}\big)-\tau_{1,0}w_{1,0}\\
&=\gamma_{1,0}\tau_{2,0} w_{2,0}-(1-\gamma_{1,0})\tau_{1,0} w_{1,0}.
\end{align*}

\end{proof}

\begin{proof}[{\bf Proof of Proposition \ref{bad3}}]
We can check that assumptions in Proposition \ref{corollary1} holds (because condition (\ref{linear2assumption0}) becomes  $(1-\tau_{2,1})A_2>(1-\tau_{1,1})A_1>0$).

Then, by apply Proposition \ref{corollary1}, we can find the aggregate output
\begin{align}Y_1=&\Big(1-\tau_{2,1}\gamma_{c,1}\Big)(1+a_2\gamma_{d,0}T_0)A_2
\dfrac{\frac{\beta_2\big[(1-\tau_{2,0})w_{2,0}+\gamma_{2,0}T_0\big]}{(1+\beta_2)}+\frac{\beta_1\big[(1-\tau_{1,0})w_{1,0}+\gamma_{1,0}T_0\big]}{(1+\beta_1)}}{1+(\frac{\gamma_{2,1}}{1+\beta_2}+\frac{\gamma_{1,1}}{1+\beta_1})\frac{\tau_{2,1}}{(1-\tau_{2,1})}}\\
&=(1-\tau_{2,1}\gamma_{c,1})A_2\frac{\beta}{1+\beta}\frac{w_{1,0}+w_{2,0}-(1-\gamma_{1,0}-\gamma_{2,0})T_0}{1+\frac{\gamma_{1,1}+\gamma_{2,1}}{1+\beta}\frac{\tau_{2,1}}{1-\tau_{2,1}}}\\
&\leq (1-\tau_{2,1}\gamma_{c,1})A_2\frac{\beta}{1+\beta}\big(w_{1,0}+w_{2,0}-(1-\gamma_{1,0}-\gamma_{2,0})T_0\big)\\
&\leq A_2\frac{\beta}{1+\beta}\big(w_{1,0}+w_{2,0}-(1-\gamma_{1,0}-\gamma_{2,0})T_0\big)\\
&<A_1 \frac{\beta}{1+\beta}\big(w_{1,0}+w_{2,0}\big) \quad \quad \text{ (because $A_2<A_1$)}\\
&=\Big(\frac{\beta_2}{1+\beta_2}w_{2,0}+\frac{\beta_1}{1+\beta_1}w_{1,0}\Big)\max(A_1,A_2)=Y^*.
\end{align}

\end{proof}

{\small 

\begin{thebibliography}{99}


\bibitem[Abdulla(2021)]{Abdulla2021} Abdulla K. (2021). Corrosive effects of corruption on human capital and aggregate productivity. \emph{Kyklos} 74: 445–462.

\bibitem[Acemoglu and Verdier(2000)]{av00}Acemoglu, D., Verdier, T. (2000). The Choice between Market Failures and Corruption. The {\it American Economic Review} 90: pp. 194-211.

\bibitem[Aghion, Akcigit, Cag, and Care(2016)]{aacc16} Aghion, P., Akcigit, U., Cag, J. and Care W. (2016). Taxation, Corruption, and Growth. \emph{European Economic Review} 86: pp. 24-51.

\bibitem[Aliprantis and Border(2006)]{ab06infinite}  Aliprantis, C. D.,  Border, K. C. (2006).  \textit{Infinite dimensional analysis}, third edition, Springer (2006).

\bibitem[d’Albis and Le Van(2006)]{ab06} d’Albis, H., Le Van, C., 2006. {Existence of a competitive equilibrium in the Lucas (1988) model without physical capital}. \textit{Journal of Mathematical Economics} 42, 46–55.



\bibitem[Bai, Jayachandran, Malesky, and Olken(2017)]{bai17} Bai, J., Jayachandran, S., Malesky, E.J.,  Olken, B.A., (2017). Firm growth and corruption: empirical evidence from Vietnam. \emph{The Economic Journal} 129: 651–677.


\bibitem[Bosi, Le Van, and Phung(2025)]{blp25} Bosi, S.,  Le Van, C.,  Phung, G.  (2025). {Economic growth with brown or green capital}. \textit{Journal of Mathematical Economics} 117, April 2025, 103101.





\bibitem[Bosi, Desmarchelier, and Ha-Huy(2022)]{bdhh22} Bosi, S., Desmarchelier, D., Ha-Huy, T. (2022) Wheels and cycles: Suboptimality and volatility of corrupted economies. \textit{International Journal of Economic Theory} 18: 440–460. 

\bibitem[Cie\'slik and Goczek(2018)]{cg18} Cie\'lik A. and  L. Goczek (2018).  Control of corruption, international investment, and economic growth - Evidence from panel data. \emph{World Development} 103: 323–335. 

	
\bibitem[Demir, Hu, Liu, and Shen(2022)]{dhls22} Demir, F., Hu, C., Liu J., and Shen, H. (2022) Local corruption, total factor productiv-ity and firm heterogeneity: Empirical evidence from Chinese manufacturing firms. \emph{Wordl Development} 151: 105-770

\bibitem[Dimaria and Le Van(2002)]{dl02} Dimaria, C-H, Le Van, C.  Optimal growth, debt, corruption and R\&D. \textit{Macroeconomic Dynamics} 6: 597-613. 


\bibitem[Doan, Vu, Tran-Nam, and Nguyen (2021)]{dvtn21}Doan, H. Q., Vu, N. H., Tran-Nam, B., Nguyen, N. A. (2021). Effects of tax administration corruption on innovation inputs and outputs: Evidence from small and medium enterprises in Vietnam. \emph{Empirical Economics} 62: 1773-1800.




\bibitem[Gourdel, Hoang-Ngoc, Le Van, and Mazamba(2004)]{levan04} Gourdel, P., Hoang-Ngoc, L., Le Van, C., Mazamba, T., 2004. {Equilibrium and competitive equilibrium in a discrete-time Lucas model}.  {\it J. Difference Equation and Applications} 10, 501–514.

\bibitem[Grundler and Potrafke(2019)]{go19} Hartwig, J. and Sturm, J.E. (2025) Corruption and economic growth: New empirical evidence, \emph{European Journal of Political Economy} 60: 101810.



\bibitem[Hartwig and Sturm(2025)]{hs25} Hartwig J. and Sturm, J.E. (2025)  Revisiting the impact of corruption on income inequality worldwide, \emph{Kyklos} 78: 206–242. 








\bibitem[Marakbi and Villieu(2020)]{mv20} Marakbi, R. and Villieu, P. (2020). Corruption, tax evasion, and seigniorage in a monetary endogenous growth model. \emph{Journal of Public Economic Theory} 22: 2019-2050. 

\bibitem[Mas-Colell, Whinston and Green(1995)]{mc95}Mas-Colell, A., Whinston, M.D., and Green, J. R., 1995.  \textit{Microeconomic theory}. Oxford University Press.

\bibitem[Mauro(1995)]{Mauro95}Mauro, P.,(1995) Corruption and growth. {\it Quarterly Journal of Economics} 110: 681–712.



\bibitem[Meon and Weill(2010)]{mw10} Meon, P. G., and Weill, L. (2010). Is corruption an efficient grease?. \emph{World Development} 38: 244259.


\bibitem[Mitra(1998)]{mitra98}
Mitra, T., 1998. {On equilibrium dynamics under externalities in a model of economic development}.  {\it Japanese Economic Review} 49, 85–107.

\bibitem[Michael(1956)]{Michael1956} Michael, E. (1956). Continuous Selections. I. \emph{Annals of Mathematics,} 63(2):361–382, 1956.



\bibitem[Mondjeli and Ambassa(2025)]{ma25} Mondjeli, I.M.M.N. and  Ambassa, M. (2025) Does corruption really matter for the structure of public expenditures ?. \emph{Structural Change and Economic Dynamics} 73: 181-195. 



\bibitem[Le Van, Morhaim, and Dimaria(2002)]{levan02} Le Van, C., Morhaim, L., Dimaria, C.-H., 2002. {The discrete time version of the Romer model}.  {\it Economic Theory} 20, 133–158.


\bibitem[Le Van and Pham(2016)]{lvp16} Le Van, C., Pham, N.-S., 2016. {Intertemporal equilibrium with financial asset and physical capital}. \textit{Economic Theory} 62: 155-199.

    

\bibitem[Leff(1964)]{Leff64} Leff, N.H. (1964) Economic development through bureaucratic corruption. {\it American Behavioral Scientist} 8: 8–14.

\bibitem[Leys(1965)]{Leys65} Leys, C. (1965) What is the problem about corruption?. {\it Journal of Modern African Studies} 3: 215–230.


\bibitem[Olken and Pande(2012)]{op12} Olken, B. A., and Pande, R., (2012). Corruption in Developing Countries. \emph{Annual Review of Economics,} 4: 479-509.



\bibitem[Petrova(2020)]{Petrova2020} Petrova, B. (2020), Redistribution and the Quality of Government: Evidence from Central and Eastern Europe. \emph{British Journal of Political Science} 51: 374-393.



\bibitem[Pham and Pham (2020)]{pp20} Pham, N.S. and Pham, T.K.C. (2020). Effects of foreign aid on the recipient country’s economic growth. \emph{Journal of Mathematical Economics} 86: 52-69.




\bibitem[Transparency International (2025)].Transparency International (2025). What is corruption? www.transparency.org/en/what-is-corruption.

\bibitem[Uberti(2022)]{Uberti(2022)} Uberti, L. J. (2022) Corruption and growth: Historical evidence, 1790–2010. \emph{Journal of Comparative Economics}, Volume 50, Issue 2, June 2022, Pages 321-349.




\bibitem[Wang, Danish, Zhang, and Wang(2018)]{wdzw18}  Wang Z., Danish, Zhang B., Wang, B. (2018), The moderating role of corruption between economic growth and CO2 emissions: Evidence from BRICS economies. \emph{Energy} 148: 506-513.
\end{thebibliography}

}
\end{document}